\documentclass[11pt,fleqn]{article}
\usepackage{fullpage}

\usepackage{amsmath}
\usepackage{amsfonts}
\usepackage{amssymb}
\usepackage{amsthm}

\newcommand{\ddim}{\operatorname{ddim}}
\newcommand{\diam}{\operatorname{diam}}
\newcommand{\rad}{\operatorname{rad}}
\newcommand{\MST}{\operatorname{MST}}
\newcommand{\etal}{{et al.\ }}
\newcommand{\eps}{\varepsilon}

\newtheorem{theorem}{Theorem}[section]
\newtheorem{lemma}[theorem]{Lemma}

\title{A light metric spanner}

\author{Lee-Ad Gottlieb
\\Department of Computer Science and Mathematics
\\Ariel University%
\thanks{Ariel, Israel. Email: \texttt{leead@ariel.ac.il}}
}

\begin{document}
\maketitle

\begin{abstract}
It has long been known that $d$-dimensional Euclidean point sets admit
$(1+\eps)$-stretch spanners with lightness 
$W_E = \eps^{-O(d)}$, 
that is total edge weight at most $W_E$ times the weight of the 
minimum spaning tree of the set \cite{DHN93}. 
Whether or not a similar result holds for metric spaces with low doubling dimension
has remained an important open problem, and has resisted numerous
attempts at resolution. 
In this paper, we resolve the question in the affirmative, and show
that doubling spaces admit 
$(1+\eps)$-stretch spanners with lightness 
$W_D = (\ddim/\eps)^{O(\ddim)}$.

Important in its own right, our result also implies a much faster 
polynomial-time approximation scheme for the traveling salesman problem
in doubling metric spaces,
improving upon the bound presented in \cite{BGK-12}.
\end{abstract}

\setcounter{page}{0}
\thispagestyle{empty}
\newpage

\section{Introduction}
Let $G=(V_G,E_G)$ be a metric graph, where vertices $V_G$ represent points of 
some metric set $S$, while the edge weights of $E_G$ correspond to inter-point 
distances in $S$.
A graph $R = (V_R,E_R)$ is a $(1+\eps)$-stretch {\em spanner} of graph 
$G = (V_G,E_G)$ if $R$ is a subgraph of $G$
(specifically, $V_R = V_G$ and $E_R \subset E_G$), and also
$d_R(u,v) \le (1+\eps) d_G(u,v)$ 
for all $u,v \in G$. 
Here, $d_G(u,v)$ and $d_R(u,v)$ denote the shortest path
distance between $u$ and $v$ in the graphs $G$ and $R$, respectively.

Low-stretch spanners have been the subject of intensive and broad-range
study over the past three decades. Research has focused on minimizing such
properties as construction time, vertex degree, graph diameter, and total edge weight
-- as well as possible trade-offs between these quantities -- 
in various settings such as planar graphs or Euclidean spaces 
\cite{Vaidya91, Salowe91, Soares94, AMS94, CK-95, AS97, DN-97, ArMoSm99, 
ADMSS95, DNS-95, GLN-02, BoGuMo04, DES08, ES13, DHN93, Sol11}.
Of particular interest is a remarkable result of the nineties, that $d$-dimensional 
Euclidean spaces admit $(1+\epsilon)$-stretch spanners with {\em lightness} 
$W_E = \epsilon^{-O(d)}$, meaning that the total spanner weight is at most
a factor $W_E$ times the weight of the minimum spanning tree ($\MST$) of the set
\cite{DHN93}.

The work of Gao \etal \cite{GGN-06} was the first to consider low-stretch
spanners in metric spaces of low doubling dimension -- a strictly more general setting
that Euclidean space. This immediately spawned a long and fruitful line of work,
showing that results comparable to those of Euclidean space can be obtained for
these spaces as well, for construction time, degree and diameter.
\cite{CG-06, GR08a, GR08b, S-09, GKK13, CLNS15, Sol14}.
However, the proof of lightness for Euclidean spanners relied heavily on
the properties specific to that space (in particular, the {\em leapfrog property}) 
and so its analysis does not carry over to doubling spaces.
In fact, the question of light low-stretch spanners for doubling metric spaces has
resisted multiple attempts at resolution, and the problem of proving (or disproving)
their existence has remained a central open
problem in the study of spanners. The best lightness bound known
for spanners in these metrics was $\Omega(\log n)$ \cite{S-09,ES13}.
Our contribution is in proving the following theorem:

\begin{theorem}\label{thm:main}
Given a graph $G$ representing a metric $S$, there exists a
$(1+\eps)$-stretch spanner $R$ of $G$ with weight 
$W_D \cdot w(\MST(G))$
where 
$W_D = (\ddim/\eps)^{O(\ddim)}$
and $0 < \eps < \frac{1}{2}$.
Spanner $R$ can be constructed in time 
$(\ddim/\eps)^{O(\ddim)} n \log^2 n$.
\end{theorem}

Hence we resolve the question in the affirmative.
We first prove that graphs which admit spanning trees that
are everywhere sparse also admit light low-stretch spanners
(Theorem \ref{thm:sparse} in Section \ref{sec:sparse}). 
We then show that doubling spaces can be decomposed into sets 
that are everywhere sparse, and that the light spanners on
these sparse sets can be joined together into a light 
low-stretch spanner for the original set (Section \ref{sec:doubling}).

\paragraph{Related work.}
Light spanners are known for only a handful of settings. These include
planar graphs \cite{ADDJS93,ACCDSZ-96},
unit disk graphs \cite{KPX-08},
graphs which are snowflakes of metrics \cite{GS-14},
and graphs of bounded pathwidth \cite{GH-12}
and bounded genus \cite{DMM-10}.

\paragraph{Application to traveling salesman.}
A polynomial-time approximation scheme (PTAS) for the Euclidean
traveling salesman problem (TSP) was presented by 
Arora \cite{A-98} (for all $d$-dimensional space) and Mitchell
\cite{M-99} (for the Euclidean plane). These celebrated results
were further improved upon by Rao and Smith \cite{RS98},
who noted that the solution tour may be restricted to the edges of
a light low-weight spanner for the sest. They used this observation
to improve the runtime of the Euclidean PTAS.

Building on these frameworks, Bartal \etal \cite{BGK-12} showed
that in {\em metric} spaces a $(1+\eps)$-approximate tour can be 
computed in time roughly
$n^{2^{O(\ddim)}}$. They also noted that whatever the minimum value
of $W_D$ may be, a tour can be computed in time 
$2^{(W_D 2^{\ddim(S)}/\eps)^{O(\ddim)} \log^2 \log n} n$.
Immediate from the work of Bartal and Gottlieb
\cite{BG-13} is a better runtime of
$2^{(W_D 2^{\ddim(S)}/\eps)^{O(\ddim)}} n$,
plus the time required to compute the spanner.
As a consequence of Theorem \ref{thm:main}, we conclude that 
doubling spaces admit a PTAS with runtime
$2^{(\ddim(S)/\eps)^{O(\ddim^2)}}n + (\ddim/\eps)^{O(\ddim)} n \log^2 n$.
This is a dramatic improvement over what was previously known.%
\footnote{The result attributed to \cite{BGK-12} is found in
Theorem 4.1 of the first arxiv version. The focus of \cite{BG-13} was a 
{\em linear}-time approximation scheme for Euclidean spaces 
under the word-RAM computational model, but aside from the time
required to construct a hierarchy and spanner, as well as a 
$\MST$ for small areas, the entire construction holds for metric spaces. A theorem 
to this effect will be appended to a future version of \cite{BG-13},
and is beyond the scope of this paper.}

\subsection{Preliminaries and notation}

\noindent{\bf Graph properties.}
Let the weight of an edge $e$ $(w(e))$ be its length. Let the 
weight of an edge set $E$ be the sum of the weights of its edges,
$w(E) = \sum_{e \in E} w(e)$. Similarly, the weight of any graph
$G=(V_G,E_G)$ is the weight of its edge-set, $w(G) = w(E_G)$.
It follows that the weight of a path is its path length.
For a partial graph $R = (V_R,E_R) \subset G$ (meaning $V_R \subset V_G$, $E_R \subset E_G$), 
$\diam_G(R)$ is the diameter of $R$
under the distance function $d_G$; it is the maximum distance in $G$ between
a pair of vertices also in $R$. Then $\diam_R(R)$ is just the length of the
longest path in $R$.

Let $B(u,r) \subset V_G$ refer
to the vertices of $V_G$ contained in the closed ball centered at $u \in V_G$ with radius $r$.
$B^*(u,r) \subset E_G$ is the edge set of the complete graph on $B(u,r)$.
A graph $G$ is said to be $s$-sparse if 
for every radius $r$ and vertex $v \in V$, the weight of 
$B^*(v,r) \cap E_G$ is at most $sr$. 
Let $A(u,r_1,r_2)$ be the annulus that includes all points at 
distance from $u$ in the range $[r_1,r_2]$,
and $A^*(u,r_1,r_2)$ the edge-set on the complete graph on $A(u,r_1,r_2)$.
Points within distance $r_1$ of $u$ are within the hollow
of the annulus, and not part of it.

\noindent{\bf Doubling dimension.}
For a point set $S$, let $\lambda = \lambda(S)$ 
be the smallest number such that every
ball in $S$ can be covered by $\lambda$ balls of half the
radius, where all balls are centered at points of $S$. 
Then $\lambda$ is the {\em doubling constant} of $S$,
and the {\em doubling dimension} of $S$ is 
$\ddim = \ddim(S)=\log_2\lambda$.
The following lemma states the well-known packing property of doubling spaces
(see for example \cite{KL-04}).

\begin{lemma}\label{lem:metric-pack}
If $S$ is a metric space and $C \subseteq S$ has
minimum inter-point distance $b$, then
$|C| = \left( \frac{2\rad(S)}{b} \right)^{O(\ddim)}$.
\end{lemma}

%

\noindent{\bf Point hierarchies.}
Similar to what was described in \cite{GGN-06,KL-04}, a subset of points $X \subseteq Y$
is an $r$-net of $Y$ if it satisfies the following properties:
\renewcommand{\labelenumi}{(\roman{enumi})}
\begin{enumerate}
\item Packing: For every $x,y \in X$, $d(x,y) \ge r$.
\item Covering: Every point $y \in Y$ is strictly within distance $r$ of
some point $x \in X$: $d(x,y) < r$.
\end{enumerate}
The previous conditions require that the points of $X$ be spaced out, yet
nevertheless cover all points of $Y$. A point in $X$ covering a point in
$Y$ is called a {\em parent} of the covered point; this definition allows
for a point to have multiple parents, but we shall assign to the point
a single parent arbitrarily. 
Two net-points $x,y$ of some $r$-net
are called $c$-{\em neighbors} if $d(x,y) < cr$. By the packing property
of doubling spaces, a net-point $x$ can have $c^{O(\ddim)}$ $c$-neighbors.

A hierarchy for a point set $S$ is composed of nested levels of $r$-nets,
where each level $i$ of the hierarchy is a $2^i$-net of the level $i-1$
beneath it. 
Note that the distance from an $i$-level point to all its descendants is less than
$\sum_{j=0}^{\infty} 2^{i-j} = 2 \cdot 2^i$.
By scaling, we will assume throughout this paper that the minimum 
inter-point distance in all sets is at least 1.
Then we have that top level $H:=\lceil \log_2 \diam(S) \rceil$ 
contains a single point, and the 0-level contains all of $S$. (Below, we will find it
convenient to allow for levels $i<0$, and each of these levels contains
all of $S$ as well. The bottom level is $L$) A full hierarchy for $S$
constructed in time $2^{O(\ddim)} n \log \Delta$, where $\Delta$ is
the {\em aspect ratio} of $S$, the ratio between the largest and smallest
inter-point distance in $S$ \cite{KL-04}. 
A hierarchy can also be constructed in time 
$2^{O(\ddim)} n \log n$, and in this case some isolated net-points 
are represented implicitly \cite{HM-06, CG-06}.
We can also augment the hierarchy with $c$-neighbor lists for all hierarchical
levels and points: It is known that there are only 
$c^{O(\ddim)}n$ neighbor pairs, and that given the hierarchy they may all
be discovered in $c^{O(\ddim)}n$ time and space.

An $r$-{\em semi-net} is a net that satisfies the covering property but not
the packing property.
A {\em semi-hierarchy} is composed of levels of nested semi-nets.

\noindent{\bf Net-Respecting spanning trees.}
We will say that a spanning tree $T$ is 
\emph{net respecting (NR)} relative to a
given hierarchy if for every edge $e$ in $T$, 
both of its endpoints are $i$-level net-points for $i$ satisfying
$12 \cdot 2^i \leq w(e) < 24 \cdot 2^{i+1}$.
(Recall from above that we may assume the existence of levels below $i=0$.)
The following lemma is adapted from Lemma 1.6 in \cite{GKK13}.

\begin{lemma}\label{lem:respect}
Every spanning tree $T$ can be converted into a net-respecting 
spanning tree $T'$ that visits all points visited by $T$, such that
$w(T')\leq 5 w(T).$
\end{lemma}

\begin{proof}
For every edge $e = (x,y)$ in $T$ do the following.
Let $x',y'$ be $i$-level net-point ancestors of $x,y$ respectively,
for the lowest $i$ satisfying 
$12 \cdot 2^i \leq d(x',y') < 24 \cdot 2^{i+1}$.
Replace edge $e$ with long edge $e' = (x',y')$, and also add the short edges
$(x,x')$ and $(y,y')$. The short edges have weight at most $2 \cdot 2^i$,
while by the triangle inequality
$w(e) 
\ge w(e') - 4 \cdot 2^i 
\ge 8 \cdot 2^i$.
While edge $e'$ is net-respecting, edges
$(x,x')$ and $(y,y')$ may not be. The short edges are themselves
replaced by the procedure above. This leads to a series of edge
replacements, and if the replacement adds an edge already in the
set, we will remove the duplicate.
The total weight of all short edges added by the recursive 
procedure is a geometric seriers summing to less than
$8 \cdot 2^i \le w(e')$.
As the ratio of the long edge to the original edge is at most $\frac{3}{2}$,
the total weight of all edges is $2w(e') (1+\frac{3}{2}) = 5w(e')$.
\end{proof}

Let $\MST^{NR}$ denote the 
net-respecting minimum spanning tree.
The following lemma related a local minimum spanning
tree to a global net-respecting minimum spanning true.
It is immediate from Lemma 1.11 in \cite{GKK13}, 
which made a similar claim concerning TSP tours.

\begin{lemma}\label{lem:weight}
Let $S$ be a point set equipped with a hierarchy, 
and $S' \subset S$.
Let $u$ be any point of $S$, and $r>0$ any value. Then
\begin{enumerate}
\item
$
w(\MST^{NR}(S) \cap B^*(u,r)) 
\le 14(w(\MST(B(u,r)))
$
\item
$
w(\MST(B(u,r))) 
\le 2w(\MST^{NR}(S) \cap B^*(u,4r))
+ 2^{O(\ddim)}r.
$
\end{enumerate}
\end{lemma}

We recall that the removal of points from a set can increase the
weight of its minimum spanning tree by at most a factor of 2.
Let $S'$ consist of $B(u,2^i)$ along with all hierarchical 
ancestors of these points up to level $i$.
Then $S' \subset B(u,3 \cdot 2^i)$. We conclude that 
\begin{equation}\label{eq:weight}
\frac{1}{7} w(\MST^{NR}(S') \cap B^*(u,r))
\le
w(\MST(S')) 
\le 4w(\MST^{NR}(S) \cap B^*(u,12r))
+ 2^{O(\ddim)}r.
\end{equation}

Finally, we note that for any $S$, 
$w(\MST(S)) \le |S| \diam(S)$.
Tighter bounds are in fact known for doubling spaces 
\cite{A-98, T-04, S-10}, but are not necessary for our results.

\section{Background: Hierarchical spanners}\label{sec:background}
Before presenting our results, we review the basic hierarchical spanner
introduced by Gao \etal \cite{GGN-06}, along with some small modifications
suited for our purposes. This spanner provides a theoretical framework for
our approach, although our ultimate construction will be significantly
more involved. The results presented in this section are mostly well-known.

Throughout this paper, we will make the simplifying assumption that $G$ 
is a complete graph. 
A {\em complete hierarchical spanner} $R$ for a graph $G$ on $S$ 
is derived from the hierarchy of $S$. 
For some constant $c$, 
for each levels $i$
we add an edge between every pair of $i$-level net-points with inter-point
distance at most $c2^i$.
More precisely, we add to $E_R$ an edge between the two vertices 
in $V_R$ that represent these points.
Recall that the top level of the hierarchy is $H$, and the
minimum inter-point distance in $S$ in 1.
For this construction we will
assume that the bottom level of the hierarchy is
$L: \lfloor \log (1/c) \rfloor$.
Let $M=H-L+1$ be the number of levels in the hierarchy for $S$.
To construct the complete hierarchical spanner, it suffices 
to construct a hierarchy for $S$ and compute
all relevant $c$-neighbors, in total time
$2^{O(\ddim)}n \log n + c^{O(\ddim)}n$ \cite{GR08a,GR08b}.

\begin{lemma}\label{lem:full-stretch}
Any spanner $R$ for graph $G$ on $S$ 
which satisfies that all $i$-level net-point
$c$-neighbor pairs have 0-stretch (for all $i$),
is a $(1+\frac{32}{c})$-stretch spanner 
for all of $G$, when $c \ge 24$.
\end{lemma}

\begin{proof}
Consider any two net-points $x,y$ in the bottom level $L$
of the hierarchy of $S$. Let $x',y'$ be their respective
$i$-level ancestors, where $i$ is the minimum value for
which $x',y'$ are $c$-neighbors
We have $d_G(x',y') \le c2^i$.
To derive a lower-bound on $d_G(x',y')$,
let $x'',y''$ be the $(i-1)$-level ancestors, and by assumption
$d_G(x'',y'') > c2^{i-1}$. Applying the triangle inequality, we have 
$d_G(x',y') 
\ge -d_G(x',x'') + d_G(x'',y'') - d_G(y'',y')
> c2^{i-1} - 2 \cdot 2^{i}
= (c-4)2^{i-1}$.
%
Since $x',y'$ are $c$-neighbors,
$d_R(x',y') = d_G(x',y')$.
Applying the triangle inequality we have
$\frac{d_R(x,y)}{d_G(x,y)}
\le \frac{d_R(x,x') + d_R(x',y') + d_R(y',y)}
	{-d_G(x,x') + d_G(x',y') - d_G(y',y)}
< \frac{d_G(x',y') + 2 \cdot 2^{i+1}}
	{d_G(x',y') - 2 \cdot 2^{i+1}}
< \frac{ (c-4)2^{i-1} + 2^{i+2} }
	{ (c-4)2^{i-1} - 2^{i+2} }
= \frac{c+4}{c-12}
= 1+\frac{16}{c-12}
\le 1 + \frac{32}{c}
.$
\end{proof}

The stretch of this spanner is arbitrarily low,
but it has poor lightness bounds:
In the trivial case where the $n$ points of $S$ reside
on the line at intervals of distance 1, it is easy to see that 
$w(E_R) = \Theta(n \log n) = \Theta(\log n \cdot w(\MST(S))$. 
We may consider being more stingy with
added edges by using the greedy algorithm of \cite{DN-97}
(see also \cite{GLN-02}), and we call this construction the 
{\em greedy hierarchical spanner}: We build the spanner as before, 
but first sort all $c^{O(\ddim)}n$ edges in increasing order. 
We consider each edge in turn, and before adding an
edge between $c$-neighbors, we check to see if the
stretch between them on the current partial spanner 
exceeds some threshold -- for example, it is at most 
$1+b$ for some constant $0 < b \le 1$ -- and only add the
edge if the condition is met.
%
%
In Lemma \ref{lem:greedy-stretch} below, we show that the greedy hierarchical 
spanner has $1+O(\frac{1}{c} + b)$ stretch. 

\begin{lemma}\label{lem:greedy-stretch}
A spanner $R$ for graph $G$ on $S$, 
for which all $i$-level net-point
$c$-neighbor pairs have 
$(1+b)$-stretch (for all $i$)
for any constant $b>0$, 
is a $(1+\frac{32}{c} + 6b)$-stretch spanner 
for all of $G$,
when $c \ge 24$.
\end{lemma}

\begin{proof}
The proof is similar to the proof of Lemma \ref{lem:full-stretch}.
Consider any two net-points $x,y$ in the bottom level $L$
of the hierarchy of $S$. Let $x',y'$ be their respective
$i$-level ancestors, where $i$ is the minimum value for
which $x',y'$ are $c$-neighbors:
$d_G(x',y') \le c2^i$.
As above,
$d_G(x',y') > (c-4)2^{i-1}$.
Now noting that by construction
$d_R(x',y') \le (1+b) d_G(x',y')$, 
and applying the triangle inequality,
we have
$\frac{d_R(x,y)}{d_G(x,y)}
\le \frac{d_R(x,x') + d_R(x',y') + d_R(y',y)}
	{-d_G{x,x'} + d_G(x',y') - d_G(y',y)}
< \frac{(1+b) d_G(x',y') + 2 \cdot 2^{i+1}}
	{d_G(x',y') - 2 \cdot 2^{i+1}}
< \frac{ (1+b) (c-4)2^{i-1} + 2^{i+2} }
	{ (c-4)2^{i-1} - 2^{i+2} }
= \frac{c+4}{c-12} + b\frac{c-4}{c-12}
< 1+\frac{32}{c} + 6b$.
\end{proof}

In closing this section, we note that
the stretch bounds of Lemmata \ref{lem:full-stretch} 
and \ref{lem:greedy-stretch} 
also hold for net-points of semi-hierarchies. Indeed, the
proofs use only the covering property of nets, and not 
their packing property.

\section{A light spanner for graphs with sparse spanning trees}\label{sec:sparse}

In this section, we show that graphs with sparse spanning trees admit
light low-stretch spanners. Note that we do not require the graph itself
to be sparse, merely that it admit a sparse spanning tree. 
In Section \ref{sec:pair}, we will show that we can build a light spanner
for the union of two low-stretch paths. In 
Section \ref{sec:spanning}, we will build on this result to
build a light spanner for all graphs with sparse spanning trees.

\subsection{A light spanner for pairs of close low-stretch paths}\label{sec:pair}

In this section we show that given a close pair $P,Q \subset G$ of low-stretch paths, 
we can compute a light low-stretch spanner for their union. We first need to define
nets, hierarchies and bipartite spanners for paths.

\noindent{\bf Path nets and hierarchies.}
Let $P \subset G$ be a path consisting of some vertices and edges in $G$.
An $r$-{\em path-net} for $P$ is an $r$-net for $P$ under the path distance function 
$d_P$. A {\em path-hierarchy} for $P$ is a full hierarchy under $d_P$,
with net points in levels $i = \lceil \log_2 w(P) \rceil$ 
and lower.

It is easy to see that an $r$-path-net of $P$ is an
$r$-semi-net for the point of $P$ under $d_G$:
We observe that the distance function
$d_P$ is non-contractive with respect to $d_G$;
$d_G(u,v) \le d_P(u,v)$ for all $u,v \in V$.
Hence, if a point is covered in the path-net, it
is covered under $d_G$ as well. However, the packing
property may not hold, so the path-net is only
a semi-net for $P$ under $d_G$.
It follows as well that the path-hierarchy for $P$ is a
semi-hierarchy for $P$ under $d_G$.

\noindent{\bf Bipartite path spanners.}
Let $G=(V_G,E_G)$ be a complete graph, and let $P=(V_P,E_P), Q=(V_Q,E_Q)$ 
be two paths in $G$ with stretch at most $(1+b_1)$
for some $1 \le b_1 \le c$.
Let each path be equipped with a net-hierarchy.
A {\em complete bipartite hierarchical spanner} $R$ for $P \cup Q$
contains all vertices of $V_P \cup V_Q$ and path edges 
$E_P \cup E_Q$ as well as the set of all edges 
$E \subset E_G$ connecting vertices of $V_P$ and $V_Q$. 
A {\em greedy biparite hierarchical spanner} $R$ for $P \cup Q$
contains all vertices of $V_P \cup V_Q$ and path edges 
$E_P \cup E_Q$, as well as an edge-set  $E' \subset E \subset E_G$: 
We sort the edges of $E$ in increasing order, and consider each
edge in turn. An edge of $E$ is added to $E'$ only
if the stretch between its endpoint on the current partial spanner
is large, say more than $(1 + b_1 + b_2)$ for some 
$0 < b_2 \le 1 - b_1$.
We show that the greedy bipartite spanner has favorable properties:

\begin{lemma}\label{lem:greedy-pair}
Let $P,Q \subset G$ be a pair of $(1+b_1)$-stretch
paths whose distance from each other is not greater than 
$c \cdot \min \{w(P),w(Q)\}$
(for $0 \le b_1 < 1$ and any $c \ge 24$).
Let $\Delta$ be the aspect ratio of $P \cup Q$. 
A greedy bipartite hierarchical spanner $R$ with 
parameters $c$ and $0 < b_2 \le 1-b_1$ 
for paths $P,Q$ can be constructed in time 
$O(cs (|V_P|+|V_Q|)\log (c\Delta))$,
and satisfies the following properties:
\begin{enumerate}
\item
Stretch:
$d_R(p,q) \le (1+\frac{32}{c} + 6(b_1 + b_2))d_G(p,q)$ 
for all $p \in P$ and $q \in Q$.
\item
Weight: 
$w(E') = \frac{12c^2s}{b_2} \cdot \min \{w(P),w(Q) \}$,
where $E' = E_R - (E_P \cup E_Q)$ is the set of new spanner edges not in $P$
or $Q$.
\end{enumerate}
\end{lemma}

\begin{proof}
\noindent {\bf (i)} 
This follows immediate from Lemma \ref{lem:greedy-stretch}.

\noindent {\bf (ii)} 
The proof is by a charging argument. We assume without loss of 
generality that $w(P) \le w(Q)$.
Take the maximum level $i$ for which $S$ has an $i$-level edge, 
and let the endpoints of the edge be $p \in P$ and $q \in Q$. 
Consider any point $\tilde{p} \in P$ satisfying
$d_P(p,\tilde{p}) \le r$ for $r=\frac{b_2}{4c}d_G(p,q)$,
and we will show that $\tilde{p}$
cannot have {\em any} edge incident to $Q$ in $E_R$:

First note that the proximity of $p$ to $\tilde{p}$ implies that
$\tilde{p}$ may be a $j$-level net-point only for values
$j \le \log r$. 
Now suppose by way of contraction that $\tilde{p} \in P$ has an edge to 
some $\tilde{q} \in Q$, and then by construction in must be that
$d_G(\tilde{p},\tilde{q}) \le c2^j \le cr$.
Let $R'$ be the partial spanner before the addition of the $i$-level
edges. Then the stretch guarantees of the paths and spanner, along
with applications of the triangle inequality, imply that
$
\frac{d_{R'}(p,q)}
	{d_G(p,q)}
\le	\frac{d_{R'}(p,\tilde{p}) + d_{R'}(\tilde{p},\tilde{q}) + d_{R'}(\tilde{q},q)}
	{d_G(p,q)}
\le	\frac{r + cr + (1+b_1)d_G(\tilde{q},q)}
	{d_G(p,q)}
\le	\frac{r + cr + (1+b_1)[d_G(p,q) + d_G(p,\tilde{p}) + d_G(\tilde{p},\tilde{q})]}
	{d_G(p,q)}
\le	\frac{(2+b_1)(r+cr) 
	+ (1+b_1)d_G(p,q)}
	{d_G(p,q)}
=	1+b_1 + \frac{(2+b_1)(r+cr)}{d_G(p,q)}
<	1+b_1 + \frac{4cr}{d_G(p,q)}
=	1+b_1+b_2.$
This implies that in $R'$ the stretch from $p$ to $q$ is not sufficient to
add to $E'$ an edge connecting them -- a contradiction.

It follows that for any $i$-level edge $e \in E'$ connecting $p \in P$ to $q \in Q$,
the entire path of $P$ within distance $\frac{b_2}{4c}w(e)$ of $p$ under $d_P$ 
has no lower level edges incident upon it. We charge the $i$-level edge to this 
segment of the path $P$. 

The statement now follows by noting that $p$ may have at most $3cs$ 
$i$-level edges incident upon it (all of which will be charged to the
same path segment of $P$): This is because all relevant
$i$-level path-net points of $Q$ within distance $c2^i$ of $p$.
Since the stretch in $Q$ is less than 2, a path connecting all these
points mush lie fully within distance $3c2^i$ of $p$, and by the 
the sparsity guarantee its length is less than $3cs2^i$. There can 
be at most $3cs$ $i$-level path-net points on a path of this length.


\noindent {\bf Runtime.} 
At each of $O(\log (c \Delta))$ levels $i$, we must locate for each
$i$-level path-net point $p \in P$ its $i$-level path-net 
$c$-neighbors in $Q$. As shown above, there are most $3cs$ such neighbors. 
Further, if two points $p \in P$ and $q \in Q$ are $i$-level $c$-neighbors, then so
are their respective parents $p' \in P$ and $q' \in Q$:
$d_G(p',q') 
\le d_G(p,p') + d_G(p,q) + d_G(q,q')
\le 2 \cdot 2^{i+1} + c2^i
= (\frac{c}{2} + 2)2^i+1
< c2^{i+1}$. 
So to compute all $c$-neighbors, it suffices to iterate down the hierarchy,
maintaining for each $i$-level net-point $p \in P$ a list of all $c$-neighbors 
in $Q$. The lists for the children of $p$ can be found by considering all children
of $p$'s $c$-neighbors, for a total of $O(cs)$ candidates.
\end{proof}

\subsection{Extension to graphs with sparse spanning trees}\label{sec:spanning}

The results in the previous section apply only to pairs of low-stretch paths. 
We will use them to obtain similar results for all graphs with sparse spanning trees.
Our plan is as follows: We first show how to {\em decompose} a spanning tree into 
paths without a stretch guarantee (Lemma \ref{lem:decomp}).
We then show how to {\em replace} any path with a light set of 
low-stretch paths (Lemma \ref{lem:replace}). Finally, we build the bipartite spanner
of the previous section on all pairs of low-stretch paths, and show that this gives
a light low-stretch spanner tree for the full set (Theorem \ref{thm:sparse}).
To do this, we will need to demonstrate that the union of all the low-stretch paths 
is still relatively sparse, and also that the union of the path-nets for these paths
can serve as a semi-hierarchy for the entire space.

\begin{lemma}\label{lem:decomp}
Given an $s$-sparse spanning tree $T = (V_T,E_T) \subset G$ of $n$ nodes with aspect 
ratio $\Delta$, $T$ may be decomposed in time $O(s |V_T| \log \Delta)$ into a set
$\cal{Q}$ of paths, with the following properties:
\begin{enumerate}
\item
Diameter:
At least one path $P_i \in \cal{Q}$ has length $\diam_T(T)$.
\item
Proximity to path-net points:
Every vertex in $V_G$ is within distance $b$ (under $d_T$)
of some path $P_j \in \cal{Q}$ of length $w(P_j) \ge b$, 
for all values $0 < b \le \diam_T(T)$ .
\end{enumerate}
\end{lemma}

\begin{proof}
The decomposition procedure removes from the spanning tree the longest path --
of length $\diam(P)$ -- and places the path in the collection $\cal{Q}$.
That is, the edges of the path are removed from $E_T$, and then all vertices
with no edges incident upon them are removed from $V_T$. After the 
removal of the longest path, a number of disjoint subtrees may remain,
and each is decomposed recursively until $V_T$ is empty.
This completes the description of the decomposition procedure.

We note trivially that the distance under $d_T$ from a removed path to all 
vertices in one of the remaining spanning trees is not greater than the
graph diameter of that spanning tree.

\noindent {\bf (i)} 
The first item follows by construction. 

\noindent {\bf (ii)} 
Take any vertex
and the paths that were removed from the spanning tree containing the vertex
at each iteration. The paths are necessarily in decreasing order of length.
Consider the first
path of length at most $b$. Then the diameter of the subtree that 
contained this path before its removal was at most $b$, 
hence the the vertex is within distance 
$b$ of the next path, which has length greater than $b$.

\noindent {\bf Runtime.} 
First note trivially that longest path of a tree can be 
found in $O(n)$ time. Now consider any subtree formed by the removal of a
path $P$ of length $w(P)$. After $O(s)$ further decomposition steps, 
the longest path in any remaining subtrees cannot be longer than
$\frac{w(P)}{2}$. This is because each decomposition step removes the longest
path of the current subtree, and leaves behind it a group of subtrees
all within distance $w(P)$ of $P$. Then the fact that the 
original subtree is $s$-sparse implies that it may contain at most
$O(s)$ paths of length greater than $\frac{W}{2}$. It follows that
the total runtime of the decomposition procedure is $O(s |V_P| \log \Delta)$.
\end{proof}

It remains to replace each of the paths in $\cal{Q}$ with a set of low-stretch paths.
Consider the following procedure (motivated by the greedy algorithm)
which takes a path $P$ and value $c$, and replaces $P$ with a set of 
low-stretch paths $\cal{P}$.
First consider the endpoint $p,q \in P$. If the endpoints are close together -- 
$d_P(p,q) \le \frac{1}{3} \diam_G(P)$ -- find a point $r$ satisfying
$d_G(p,r), d_G(q,r) \ge \frac{1}{3} \diam_G(P)$, segment $P$ at
$r$, and continue the procedure below separately on the two smaller paths. 
Otherwise, continue the procedure on $P$ itself:

Iteratively create a path-net for $P$, at each iteration 
$i= L,L+1,\ldots$ promoting some of the $(i-1)$-level path-net points to also be $i$-level
path-net points: We begin by promoting the first point in $P$ -- call
this $p$ -- and proceeding down the path promote in turn 
every $(i-1)$-level point at path distance $2^i$ or greater from the
previously promoted point. 

Then beginning again at $p$, we search down the path for the {\em associate} of $p$, 
the first $i$-level path-net point $q$ satisfying $d_G(p,q) \ge c2^i$.
Now, if the stretch $\frac{d_P(p,q)}{d_G(p,q)}$ is greater than 
$1+\frac{1}{3c(c+1)s}$, 
remove the partial path connecting $p,q$, and replace it by a directed edge. 
The partial path is then segmented into smaller paths at its
$(i-2)$-level path-net points, and these small paths are all added to $\cal{P}$.
If the partial path is removed, the procedure continues with point $q$, 
searching down the path for its associate. 
If the partial path is not removed, then the procedure takes the first
$i$-level path-net point $r$ following $p$, and searches down the path for
$r$'s associate. Iteration $i$ terminates upon reaching an $i$-level
path-net point $r$ with no associate. In this case, we locate the final
$i$-level path-net point in $P$ -- call this $q$ -- and remove the partial
path connecting $p$ and $q$, replacing it by a direct edge. The
removed partial path, as well as the path beyond $q$, are decomposed into
smaller paths at their $(i-2)$-level path-net points, and added to $\cal{P}$.
The procedure terminates upon reaching a level $i$ for which the endpoint
$p$ has no associate.

Lemma \ref{lem:replace} below shows that the final path set $\cal{P}$
has favorable properties, which we will use in the construction
of the spanner.

\begin{lemma}\label{lem:replace}
Let $P=(V_P,E_P)$ be an $s$-sparse path in $G$ with arbitrary stretch
and aspect ratio $\Delta$.
For any constant $c \ge 24$, the collection $\cal{P}$ may be computed in time
$O(cs |V_p| \log \Delta)$, and possesses the following properties:
\begin{enumerate}
\item
Vertex cover: 
The union of vertices in all sets $P_i \in \cal{P}$ is exactly $V_P$
($\cup_i V_{P_i} = V_P$), and 
$\sum_i |V_{P_i}| \le 3|V_P|$.
\item
Stretch: 
Each path in $\cal{P}$ has stretch at most $(1+\frac{32}{c})$.
\item
Diameter: At least one path $P_i \in \cal{P}$ has length 
$\frac{\diam_G(P)}{4} \ge \frac{\diam_P(P)}{4s}$ 
or greater.
\item
Proximity to path-net points:
Every vertex $v \in V_P$ of is within $d_G$ distance
$16(c+1)b$ of some path $P_j \in \cal{P}$ of weight 
$w(P_j) \ge b$, where $b$ is any value satisfying 
$0 < b \le \frac{\diam_G(P)}{4}$.
\item
Path sparsity:
Each path in $\cal{P}$ is $3s$-sparse.
\item
Graph sparsify and weight:
The union of all paths $P_i \in \cal{P}$
form a graph that is both $3s(3cs+1)$-sparse and has total weight 
at most $3s(3c(c+1)s+1) \cdot w(P)$.
\end{enumerate}
\end{lemma}

\begin{proof}

\noindent {\bf (i)} 
This covering immediate from the construction, and the number 
of vertices follows by noting that a removed path may duplicate
the path endpoints, but removes at least a single vertex from
the original path.

\noindent {\bf (ii)} 
Let $p,q$ be any pair of $i$-level path-net points in a path 
$P' \in \cal{P}$, and we show that 
the path stretch of $p,q$ is less than $1+\frac{1}{c}$.
Then the item will follow from Lemma \ref{lem:greedy-stretch}.

Since $p,q$ are found on the same path, it must be that 
all $(i+1)$-level associate pairs $r,t$ straddling the 
path segment from $p$ to $q$ have path stretch less than
$1+\frac{1}{3c(c+1)s}$, as otherwise the path connecting $r$ and $t$
would have been removed and segmented along its $(i-1)$-level 
path-net points, and $p,q$ could not be in the same path
segment in $\cal{P}$. Recall by construction the $d_G$-distance from $r$
to $t$ is in the range $[c2^{i+1},(c+1)2^{i+1})$, and since
all paths of $\cal{P}$ are $3s$-sparse (item (iv) of 
the lemma), the $d_{P'}$ distance 
from $r$ to $t$ is at most $3s \cdot (c+1)2^{i+1}$. 
Were the pair $p,q$ to have path stretch at least $1+\frac{1}{c}$,
then $r,t$ would have path stretch at least 
$ 	\frac{d_{P'}(r,t)}{d_G(r,t)}
= 	\frac{d_{P'}(r,p) + d_{P'}(r,t) + d_{P'}(t,p)}{d_G(r,t)}
\ge 	\frac{d_G(r,p) + (1+\frac{1}{c})d_G(r,t) + d_{P'}(t,p)}{d_G(r,t)}
\ge 	\frac{d_G(r,t) + \frac{2^i}{c})d_G(r,t)}{d_G(r,t)}
\ge	\frac{3s \cdot (c+1)2^{i+1} + 2^i/c}{3s \cdot (c+1)2^{i+1}}
= 	1 + \frac{1}{3c(c+1)s}$,
which is a contradiction.

\noindent {\bf (iii)} 
The endpoints $p,q \in P$ were chosen so that $\diam_P(P) \ge \frac{\diam_G(P)}{3}$.
By construction, the final path remaining after all iterations retains $p$ as
an endpoint. On the other end, at each iteration $i$ a small end-segment of the path
may be segmented, but these account for less than a fraction
$\frac{1}{c}\sum_{i=0}^\infty 2^{-i} = \frac{2}{c} \le \frac{1}{12}$ of the total
path length. So the diameter is at least 
$\frac{\diam_G(P)}{3} \cdot \frac{11}{12} > \frac{\diam_G(P)}{4}$.
Finally, the $s$-sparsity of path $P$ directly implies that
$\diam_P(P) \le s\diam_G(P)$.

\noindent {\bf (iv)} 
For every added edge connecting $i$-level path-net points $p,q$, 
the removed path was all within distance $c2^i$ of $p$ (under $d_G$)
of the new edge connecting $p,q$, which have length in the range 
$d_G(p,q) \in [c2^i,(c+1)2^i)$.
Now the removed path may itself have had paths removed from
it at earlier stages of the iteration, but all these
paths must be within distance
$\sum_{j=0}^i c2^j < c2^{i+1} < 2d_G(p,q)$ of the new edge.
Also note that when this path was removed, it was decomposed into 
smaller paths of path-length at least 
$2^{i-2} = \frac{2^{i+1}}{8}$.
Hence, the length of the edges feature jumps of at most
$8(c+1)$, and so every vertex is within distance 
$16(c+1)b$ of an edge of length at least $b$.

\noindent {\bf (v)} 
Take an edge added by the procedure between points $p,q$. 
Then in the original path $P$ there must have been a heavier
path connecting $p,q$, and by the arguments in item (iv) 
this path was fully within distance
$d_G(p,q)$ of both $p$ and $q$. Any ball containing the new 
edge must be of radius at least $\frac{d_G(p,q)}{2}$, and so
it contains the entire (removed) old path within 3 times its 
radius. So the weight of the smaller ball is less than the weight
of the larger larger ball, and the sparsity is $3s$.

\noindent {\bf (vi)}
For the claim of weight:
We have already noted in item (iii) that the edges added when removing 
an end-segment sum to much less than the weight of the path.
For the other edges, the proof proceeds by a charging argument. 
For an edge $e$ added
between points $p,q$ in iteration $i$, charge its weight to the 
edges of the removed path. Each edge $e'$ in the removed partial
path $Q$ receives change $w(e) \frac{w(e'}{w(R)}$, and since $Q$
must have weight at least $w(e)(1+\frac{1}{3c(c+1)s})$, the charge to $e'$
is at most $\frac{w(e')}{1+\frac{1}{3c(c+1)s}}$. Now, if $e'$ itself
was added in an earlier iteration $j$, then the new charge to $e'$ 
due to the addition of $e$ is itself passed down to earlier removed
paths, until all charges lie solely on edges present in the original
path $P$. Let $e$ be an edge of $P$, and is follows that the sum
of charges placed on $e$ (included the cost of $e$ itself) is less than
$w(e) \sum_{i=0}^\infty (1+\frac{1}{3c(c+1)s})^{-i} < (3c(c+1)s+1)w(e)$.

For the claim of sparsity:
It follows from the proof of items (iii) and (iv) that all edges 
of $\cal{P}$ inside an $r$-radius ball charge to edges of the 
original path $P$ within distance $3r$ of the ball's center. 
As $P$ is $s$-sparse, the total weight in the small ball is
$3s(3c(c+1)s+1)$, and the total weight of all edges in $\cal{P}$ is
$3s(3c(c+1)s+1) \cdot w(P)$.

\noindent{\bf Runtime.}
There are $O(\log \Delta)$ iterations, and at each one we can make a single pass
on the path, promoting net-points and recording their order in a list.
Then we must discover for each $i$-level path-net point $p$ its 
$c$-neighbors. But these are all within path distance $O(cs2^i)$ of 
$p$, so it suffices to inspect only the next $O(cs)$ net-points
in the list. The runtime follows.
\end{proof}

Joining together Lemmata \ref{lem:greedy-pair}, \ref{lem:decomp}
and \ref{lem:replace}, we can prove that spaces with sparse spanning
trees admit light low-stretch spanners:

\begin{theorem}\label{thm:sparse}
Let $T = (V_T,E_T)$ be an $s$-sparse spanning tree of a complete graph $G$
with aspect ratio $\Delta$. 
Then in time $(s/\eps)^{O(1)} \cdot |V_T| \log \Delta$ 
we can construct for $G$ a spanner $R$ with
\begin{enumerate}
\item
Weight: $W_s \cdot w(T)$ for $W_s = (s/\eps)^{O(1)}$.
\item
Stretch: $(1+\eps)$.
\end{enumerate}
\end{theorem}

\begin{proof}
The construction is straightforward: Decompose the spanning tree $T$
into paths using the {\em decomposition} procedure of Lemma \ref{lem:decomp},
and then replace each path with a set of low-stretch paths using
the {\em replacement} procedure of Lemma \ref{lem:replace}
with parameter 
$c_1 = 3 \cdot 2^7 / \eps$. 
For every pair of paths
in the new set, build the greedy bipartite spanner of Lemma \ref{lem:greedy-pair}
with parameters 
$c_2 = \max\{16c_1+18, 8s \} \cdot 2^6 / \eps$
and 
$b_2 = \eps/12$. 
The new graph is $G'$.
Recall that the bipartite spanner will only add edges if the path pair is 
sufficiently close.

\noindent{\bf (i)}
We first must prove sparsity: Tree
$T$ is $s$-sparse, and the sparsity condition is not affected 
by the decomposition. Take any path $P$ the decomposition, and
$P$ is replaced by a set of paths connecting the vertices of $P$
with sparsity $O(c_1^2s)$ (Lemma \ref{lem:replace}(vi)).
Hence, before the addition of the bipartite edges
$G'$ possessed sparsity $s' = O(c_1^2s^2)$. 

Returning to proof of weight,
it follows from Lemma \ref{lem:replace}(vi) that the weight of the
paths of $G'$ before the construction of the bipartite spanners
was $3s(3c_1(c_1+1)s+1) \cdot w(T)$.
A bipartite spanner with parameter $c_2$ is built for each pair of paths,
and we will charge the cost of the new edges to the shorter path.
By sparsity, any path $P$ is within distance $c_2 \cdot w(P)$ of 
$O(c_2s') = O(c_1^2 c_2 s^2)$ 
longer paths. Charge $P$ for the additional edges added by the 
bipartite spanner to each of these paths. By Lemma \ref{lem:greedy-pair} 
this yields a charge of 
$O(c_2^2 s/b_2) \cdot w(P)$ per spanner, for a total charge of 
$O(c_1^2 c_2^3 s^3 / b_2) \cdot w(P)$. It follows that
$w(G') = O(c_1^2 c_2^3 s^3 / b_2) \cdot w(T)$.

\noindent{\bf (ii)}
By Lemma \ref{lem:replace}, the intra-path stretch in all paths
produced by the replacement procedure is $1+\frac{32}{c_1}$.
By construction of the bipartite spanner, the stretch among all 
$i$-level path-net points that are $c_2$-neighbors is 
$\max \{ 1+b_2, 1+\frac{32}{c_1} \}$. 

We require a semi-hierarchy:
Below we show that each vertex $v$ in $G'$ is within distance
$(16c_1+18)2^i$ of some $i$-level path-net point for all levels
$i \le \lceil \diam_T(T)/8s \rceil$.
It follows that if we promote all $i$ level path-net points 
(for each $i$) to be $j$-level points for 
$j = i+ \log \max\{16c_1+18, 8s \}$, 
then each vertex $v$ is within distance at most $2^j$ of some 
$j$-level point, and this hold for all levels
$j \le \lceil \diam_T(T) \rceil = H$.
Hence we have a valid semi-hierarchy for all points.
Substituting $j$ into the stretch bound above, we have that 
the the stretch among all $j$-level semi-hierarchy net-points that are
$\frac{c_2}{\max\{16c_1+18, 8s \}}$-neighbors is
$\max \{ 1+b_2, 1+\frac{32}{c_1} \}$. 
Substituting in the values of $c_1,b_2$ chosen above, we have that
the the stretch among all $j$-level semi-hierarchy net-points that are
$(64/\eps)$-neighbors is
$1+\frac{\eps}{12}$.
Then by Lemma \ref{lem:greedy-stretch}, the resulting spanner has stretch 
$1 + \frac{32}{64/\eps} + 6\frac{\eps}{12}
= 1 + \eps$.

It remains only to verify the above claim,
that each vertex in $G'$ is within distance 
$(16c_1+18)2^i$
of some (not yet promoted) $i$-level path-net point:
After the decomposition procedure, 
for any value $0 < b \le \diam_T(T)$ 
each vertex $v \in V_G$ is within
distance $b$ of some vertex $w \in P_i$ for which $w(P_i) \ge b$
(Lemma \ref{lem:decomp}(ii)). 
After the replacement procedure, $w$ is within distance 
$16(c_1+1)b'$ of some path of length at least $b'$ where
$0 < b' \le \frac{\diam_P(P_i)}{4s}$. 
Setting $b' = b$, we have that $v$ is within distance 
$(16c_1+17)b$ of some path of length
at least $b$, where $0 < b \le \frac{\diam_T(T)}{4s}$. 
A path of this length possesses a path-net point at level 
$i = \lceil \log (b/2) \rceil$ or greater.
It follows that $v$ is within distance
$(16c_1+18)2^i$ of some $i$-level path-net point for all levels
$i \le \lceil \diam_T(T)/8s \rceil$.

\noindent{\bf Runtime.}
By Lemmata \ref{lem:decomp} and \ref{lem:replace}, 
the decomposition and replacement steps can be done in time 
$O(c_2 s |V_T| \log \Delta)$.
To compute each bipartite spanners, 
for each path $P$ in $G'$ we must 
have all paths of length $w(P)$ or greater within 
distance $c_2 \cdot w(P)$ of $P$. 
Recall from the proof of {\bf (i)} above
that the sparsity of $G'$ before the addition
of bipartite spanner edges is $s' = (s/\eps)^{O(1)}$.

After the initial decomposition step, assign a path-hierarchy to
each path, and compute all $i$-level $2c_2$-neighbors among all
paths and pairs of paths, for all $i$. This can all be done in time 
$(s/\eps)^{O(1)} n \log \Delta$. (The construction and analysis
is identical to that of Lemma \ref{lem:decomp}.) For each path
$e$ created by the replacement step, we compute its distance
to all $i$-level path-net points in the original edge within 
distance $c_2 2^i$.
Since the current graph is $s'$-sparse and has aspect ratio $\Delta$
this can be done in time 
$(s/\eps)^{O(1)} n \log \Delta$ 
over all paths,
if we precomputed for each vertex its closest $i$-level 
path-net point for all $i$.
For each $i$-level path-net point, we maintain a list of the
replacement path which discovered this point, sorted in order
of path length. Note that we may sort all replacement paths once,
and then compute close path-net point for each one in turn.

For each path $P$ and each $i \le \log(w(e))$, we inspect 
the $i$-level path-net points within distance $c_2 2^i$,
and all $i$-level $c_2$-neighbors of those path net points,
and their listed paths of length greater than $w(P)$, and
this suffices to discover for $P$ all longer edges within
distance $c_2 \cdot w(P)$. The runtime follows from the sparsity
and aspect ratio of the current graph.
\end{proof}

\section{A light spanner for general metric graphs}\label{sec:doubling}
In Section \ref{sec:sparse} we showed how to construct a light
spanner for spaces that have sparse spanning trees. In this section
we complete the proof of Theorem \ref{thm:main} by showing how
to decompose general metric graphs of low doubling dimension 
into graphs with sparse spanning
trees. Then the light spanners for the decomposed sparse spaces can all
be joined into a single light spanner for the original metric space.
Our approach uses a technique developed for computing
near-optimal traveling salesman tours in polynomial time
\cite{BGK-12,BG-13}, 
although our setting is much less restrictive, and so the problem 
of finding a good decomposition is simpler.

For some graph $G$ equipped with a hierarchy, define $F(u,i)$ for $i$-level
vertex $u$ to include all points of $B(u,i) \cap G$, as well as all their
hierarchical ancestors up to level $i$. The following preliminary lemma shows
that we can spin off a (slightly) dense area of the graph:

\begin{lemma}\label{lem:spinoff}
Let $G$ be a graph with aspect ratio $\Delta$, equipped with
a hierarchy. For some fixed $f = (\ddim/\eps)^{O(\ddim)}$, let
$i$ be the lowest level for which there exists some $u$ satisfying
$\MST^{NR}(F(u,i)) > f2^i$. Then $G$ may be segmented into intersecting
subgraphs $G',D \subset G$ with the following properties, for any value $c \ge 2$:
\begin{enumerate}
\item
Weight in $G'$: 
$w(\MST^{NR}(F(u,i) \cap G')) \le \frac{f}{4} 2^i$,
while $w(\MST^{NR}(G)) - w(\MST^{NR}(G')) \ge \frac{f}{8}2^i$.
\item
Sparsity of $D$: $\MST(D)$ is $c^{O(\ddim)} f$-sparse, while $D$ has diameter $O(c) \cdot 2^i$.
\item
Neighbor proximity:
For every $c$-neighbor point pair in the hierarchy of $G$,
the pair is found together in $A$ or $B$ or in both.
\end{enumerate}
\end{lemma}

\begin{proof}
Subset $D$ includes all points in $B(u,(13+c) \cdot 2^i)$, 
along with all their net-points up to level $i$. 
Subset $G'$ includes all points of $G - B(u,13 \cdot 2^i)$, 
as well as all $i$-level (and higher level) net-points in 
$B(u,13 \cdot 2^i)$. Below, we will also add some more points to $A$. 

\noindent {\bf (ii)} 
Recall that by assumption each $(i-1)$-level point $v$ satisfies
$\MST^{NR}(F(v,i-1)) \le f2^{i-1}$. 
A spanning tree for any ball of radius $2^i \le r \le (13+c) \cdot 2^i$ 
in $D$ can be formed by covering that ball with the spanning trees of
$r^{O(\ddim)}$ sets $F(v,i-1)$, and then connecting the 
centers of these sets at an additional cost of $r^{O(\ddim)} 2^i$,
for a total weight of $r^{O(\ddim)}f 2^i$. By the optimally of 
$\MST(D)$, its weight inside the given ball cannot be greater than this.
A similar result holds for values of $r$ less than $2^i$, covering
with sets $F(v,k)$ for $k < i-1$.

\noindent {\bf (iii)} 
Clearly $c$-neighbor pairs of levels $i$ or higher are both found in
$G'$. For levels $k<i$, if one of the points is not found in $G'$, then by
construction it is found in $B(u,13 \cdot 2^i)$, and so the second point is found in 
$B(u,13 \cdot 2^i + c 2^k) \subset B(u,(13+c) \cdot 2^i) = D$.

\noindent {\bf (i)} 
We first discuss the weight of $\MST^{NR}(G)$ in the area around $u$:
A consequence of Equation \eqref{eq:weight} is that 
$B^*(u,12 \cdot 2^i)$ 
had intersected edges of $\MST^{NR}(G)$
weighing more than $\frac{f}{4} 2^i$.
At the same time, the larger ball
$B(u,13 \cdot 2^i)$ admitted 
a minimum spanning tree of weight 
$2^{O(\ddim)} f 2^i$
(by arguments identical to the proof of item {\bf (ii)} above). 
So by Lemma \ref{lem:weight}(i) 
$B(u,13 \cdot 2^i)$
intersected edges of $\MST^{NR}(G)$ of weight $2^{O(\ddim)} f 2^i$.

We now add some more points to $G'$:
Set $j=i-a \log \ddim$ for some absolute constant $a>2$ to be specified below.
Add to $G'$ all $j$-level points in $B(u,13 \cdot 2^i)$.
Further, let radius $r \in [12 \cdot 2^i + 72 \cdot 2^j, 13 \cdot 2^i - 72 \cdot 2^j]$ 
be such that
$w(\MST^{NR}(G) \cap A^*(u,r- 72 \cdot 2^j, r+ 72 \cdot 2^j)
\le \frac{1}{4} w(\MST^{NR}(G)) \cap B^*(u,r- 72 \cdot 2^j))$ --
that is, the weight of $\MST^{NR}(G)$ inside the annulus is a fraction
of that inside the hollow.
(Such a value $r$ must in fact exist: There are 
$\Theta(2^i/2^j) = \Theta(2^a \ddim)$ non-intersecting annuli.
If each contained edges of $\MST^{NR}(G)$ weighing a fraction at least 
$\frac{1}{4}$ of the edge-weight within their hollow, the total edge-weight
within $B(u,13 \cdot 2^i)$ would be 
$(1+\frac{1}{4})^{2^a \ddim} f 2^i$;
for a sufficiently large choice of $a$, this exceed the upper-bound proved above.)
Add to $G'$ all points of the smaller annulus 
$A(u,r- 24 \cdot 2^j, r+ 24 \cdot 2^j)$, 
as well as all points outside outer edge of the annulus --
in short, all points outside $B(u,r- 24 \cdot 2^j)$.

To prove the first part of item {\bf (i)}, note that 
$F(u,i) \cap G'$
contains only points of level $j$ or higher, and these can
connected by an MST of weight 
$(2^i/2^j)^{O(\ddim)} 2^i = \ddim^{O(2^a \ddim)} 2^i \ll \frac{f}{20} 2^i$.
The final inequality follows when $f$ is chosen to be sufficiently 
large with respect to $a$. Then by Lemma \ref{lem:respect}, 
a connecting net-respecting MST has weight less than $\frac{f}{4} 2^i$.

To prove the second part of item {\bf (i)}, we will construct a 
net-respecting spanning graph for $G'$ with smaller weight than 
$\MST^{NR}(G)$. Recall that $G$ and $G'$ differ only on points within
$B(u,r -24 \cdot 2^i)$.
We add to the spanning graph of $G'$
all edges of $\MST^{NR}(G)$ fully outside of $B(u,r + 24 \cdot 2^i)$
(and this is outside the smaller annulus 
$A(u,r- 24 \cdot 2^j, r+ 24 \cdot 2^j)$).
We also add to $G'$ those edges with only a single endpoint inside $B(u, r + 24 \cdot 2^i)$,
and these net-respecting edges must be incident on the annulus
$A(u,r- 24 \cdot 2^j, r+ 24 \cdot 2^j)$ if they are shorter than $48 \cdot 2^j$,
or on $j$-level points in $(u, r -24 \cdot 2^j)$ if they are longer.
We then claim that the points of $G'$ inside the ball 
$B(u, r +24 \cdot 2^j)$ 
can be covered by a graph whose weight is lighter than that of 
$\MST(G) \cap B^*(u,r+24 \cdot 2^i)$ by at least $\frac{f}{8} 2^i$,
from which the second part of the item follows.

First note that a graph connecting all $j$-level points
in $B(u,13 \cdot 2^i)$ has weight much less than $\frac{f}{80} 2^i$
for appropriately chosen $f$ (as above in the proof of the first part of {\bf(i)}), 
and so by Lemma \ref{lem:respect} a net-respecting MST connecting
only these points has weight much less than $\frac{f}{16} 2^i$.
Further, all points of the annulus
$A(u,r- 24 \cdot 2^j, r+ 24 \cdot 2^j)$ 
can be connected to some $j$-level point at total cost at most
$w(\MST^{NR}(G) \cap A^*(u,r- 72 \cdot 2^j, r+ 72 \cdot 2^j))$
(which was chosen above to be at most
$\frac{1}{4} w(\MST^{NR}(G) \cap B^*(u,r- 72 \cdot 2^j))$):
The annulus cuts $w(\MST^{NR}(G))$ into disjoint paths. For each
such path, if it is incident on a $j$-level (or higher level)
point then we are done. Otherwise, consider the tail of this path,
the final segment exiting the smaller annulus
$A(u,r- 24 \cdot 2^j, r+ 24 \cdot 2^j)$. Since $\MST^{NR}(G)$ is
net-respecting, and the path touches only net-points of level less
than $j$, the tail must be of length at least $24 \cdot 2^j$, since
it cannot jump immediately to a point outside the larger annulus.
Cut off this tail, and instead add to the path a much shorter tail
of length at most $2 \cdot 2^j$ connecting to the closest $j$-level point.
So the total cost of the spanning graph for 
$G' \cap B(u, r +24 \cdot 2^j)$
is at most
$\frac{f}{16} 2^i + \frac{1}{4} w(\MST^{NR}(G) \cap B^*(u,r- 72 \cdot 2^j))$,
while 
$\MST^{NR}(G) \cap B^*(u, r +24 \cdot 2^j)$
is of course at least 
$w(\MST^{NR}(G) \cap B^*(u,r- 72 \cdot 2^j))$,
a difference of 
$\frac{3}{4} w(\MST^{NR}(G) \cap B^*(u,r- 72 \cdot 2^j)) - \frac{f}{16} 2^i
\ge	\frac{3}{4} w(\MST^{NR}(G) \cap B^*(u,12 \cdot 2^i)) - \frac{f}{16} 2^i
\ge	\frac{3}{4} \cdot \frac{f}{4} 2^i - \frac{f}{16} 2^i
=	\frac{f}{8} 2^i
$.
\end{proof}

By using the technique of Lemma \ref{lem:spinoff} to repeatedly spin off
relatively sparse areas of the graph, we can decompose the entire graph into
relatively sparse areas, as in the following theorem:

\begin{theorem}\label{thm:doubling}
Any graph $G$ with aspect ratio $\Delta$ and equipped with a hierarchy, 
may be segmented into a set of subgraphs
$\cal{D}$ ($\cup_i D_i = G$) with the following properties, for any $c \ge 2$ and for
$f$ as above:
\begin{enumerate}
\item
Sparsity: For all $D_i \in \cal{D}$, $\MST(D_i)$ is $c^{O(\ddim)} f$-sparse.
\item
Weight: $\sum_i w(\MST(D_i)) = c^{O(\ddim)} w(\MST^{NR}(G))$
\item
Neighbor Proximity: 
For every $c$-neighbor point pair in the hierarchy of $G$, the pair
is found together in at least one set $D_i \in \cal{D}$.
\item
Point occurrence:
Each point of $G$ appears in at most $c^{O(\ddim)} \log \Delta$ different sets in $\cal{D}$.
\end{enumerate}
The segmentation can be computed in time 
$c^{O(\ddim)} n \log n (\log n + \log \Delta)$.
\end{theorem}

\begin{proof}
The technique of Lemma \ref{lem:spinoff} requires knownledge of $\MST^{NR}(G)$
(to compute the annulus), so we must first build this graph.
We construct the full hierarchical spanner with parameter $c$.
This can all be done in time $c^{O(\ddim)} n \log n$, and the resulting
graph has only $c^{O(\ddim)}n$ edges \cite{GR08a, GR08b}. We then extract
from the edges only those that are net-respecting, and build $\MST^{NR}(G)$
using this set, in total time $c^{O(\ddim)} n \log n$. 
We will need to maintain $\MST^{NR}(G)$ under a long series of point deletions, 
and this can be done in total time $c^{O(\ddim)} n \log^2 n$ 
(see Theorem 6 in \cite{HLT-01}).

Beginning at the lowest hierarchical level $L$ and iterating upwards, for each 
$i$-level net-point $u$ we compute $F(u,i)$ and construct $\MST^{NR}(F(u,i))$, 
applying the segmentation technique of Lemma \ref{lem:spinoff} if necessary. 
The subset $D$ is added to $\cal{D}$, and we iterate
the procedure on subset $G'$ until there are no more heavy neighborhoods, at which
point the final subset $G'$ is added to $\cal{D}$ as well.

The sparsity of sets in $\cal{D}$ and neighbor proximity follow immediately from
Lemma \ref{lem:spinoff}(ii) and (iii), respectively. 
The bound on point occurences follows from the packing
property and height of the hierarchy, and a corollary of this bound is that
$c^{O(\ddim)} n \log n \log \Delta$ time suffices to compute all graphs $\MST^{NR}(F(u,i))$.
For the weight, let $G_i$ be the graph in the
$i$-th iteration, and $w_i = w(\MST^{NR}(G_{i+1})) - w(\MST^{NR}(G_{i}))$.
Clearly $\sum_i w_i \le \MST^{NR}(G)$. By Lemma \ref{lem:spinoff}(i) and (ii),
$\sum_i w(\MST(D_i)) = \sum_i c^{O(\ddim)} w_i = c^{O(\ddim)} w(\MST^{NR}(G))$.
\end{proof}

We can now complete the proof of Theorem \ref{thm:main}:

\begin{proof}
Fix $c = 64/\eps$.
Given graph $G$, we build a full hierarchy for the entire space, and
add to the spanner $R$ the edges of the complete hierarchical 
$(1+\frac{\eps}{12})$-stretch spanner 
of Lemma \ref{lem:full-stretch} on all levels below $i = H - \log (n^2)$. 
Note that $2^H \le \diam(S)$, and so the longest edge added by
this construction is of length 
$O(1/\eps) \cdot 2^{H-\log(n^2)} =  O(\diam(S)/\eps n^2)$. So the entire cost
of all these edges is less than $O(1/\eps) \cdot \MST(G)$.\footnote{
This observation is due to Shay Solomon, and is similar to one made
by Arora in the context of Euclidean TSP \cite{A-98}.}
We then remove from consideration all net-points below level
$i = H - \log (n^2)$, set $L=i$, and focus on the remaining
$M = H-L = O(\log n)$ levels. Note that the aspect ratio of this set
is only $\Delta = 2^{O(M)} = n^{O(1)}$.

We apply Theorem \ref{thm:doubling} to decompose the remaining graph into all
sparse sets in time $(\ddim/\eps)^{O(\ddim)}n \log^2 n$, and construct 
the $(1+\frac{\eps}{12})$-stretch
spanner of Theorem \ref{thm:sparse} for each sparse subset. 
Since any point of $G$ appears in most $(1/\eps)^{O(\ddim)}\log n$
sets (Theorem \ref{thm:doubling}(iv)) the total time to construct all these
spanners is $(\ddim/\eps)^{O(\ddim)}n \log^2 n$.
Then the stretch between any pair of $c$-neighbors is $(1+\frac{\eps}{12})$,
and it follows from Lemma \ref{lem:greedy-stretch}
that the total graph stretch is $(1+\eps)$.

The weight guarantee follows from Theorem \ref{thm:doubling}(ii)
in conjunction with Theorem \ref{thm:sparse}(i).
\end{proof}

\paragraph{Open problems.}
Our results suggest several avenues for further research. 
Can $W_D$ be reduced to $\eps^{-O(\ddim)}$, to match
what is known for Euclidean space?
Also, does the simple greedy hierarchical spanner alone yield 
a light spanner?
One can also improve the runtime of our construction.

\bibliographystyle{alpha}
\bibliography{light}

\end{document}